\documentclass[a4paper,USenglish,cleveref,numberwithinsect]{lipics-v2019}
\usepackage{hyperref}
\usepackage[T1]{fontenc}
\usepackage[textsize=footnotesize]{todonotes}
\newtheorem{observation}{Observation}

\newtheorem{assumption}[theorem]{Assumption}

\usepackage{thmtools,thm-restate}

\newcommand\blfootnote[1]{%
	\begingroup
	\renewcommand\thefootnote{}\footnote{#1}%
	\addtocounter{footnote}{-1}%
	\endgroup
}

\usepackage[ruled,linesnumbered]{algorithm2e} 

\newcommand{\YES}{\textup{\textsf{YES}}}
\newcommand{\NO}{\textup{\textsf{NO}}}
\newcommand{\Oh}{{\mathcal O}}
\newcommand{\nat}{\mathbb{N}}

\newcommand{\FPT}{\mbox{{\sf FPT}}}

\def\eg{{\em e.g.}}

\def\ie{{\em i.e.}}
\def\etal{{\em et al.}}

\newcommand{\PPP}{\mathcal{P}}
\newcommand{\QQQ}{\mathcal{Q}}

\newcommand{\bfP}{\mathbf{P}}
\newcommand{\bfQ}{\mathbf{Q}}

\newcommand{\pref}[2]{\mathbf{pre}(#1,#2)}
\newcommand{\suff}[2]{\mathbf{suf}(#1,#2)}

\newcommand{\mor}{{\sc Colored Path}}
\newcommand{\cmor}{{\sc Colored Path$^\star$}}

\newcommand{\gcmor}{{\sc Connected Obstacle Removal}}

\newenvironment{cProof}
	{\begin{proof}[Proof of Claim.] }
	{
	\end{proof}
	
	}

\newif\iflong
\newif\ifshort

\longtrue

\iflong
\else
\shorttrue
\fi

\usepackage{microtype,boxedminipage}

\title{Removing Connected Obstacles in the Plane is FPT}

\author{Eduard Eiben}{Department of Computer Science, Royal Holloway, University of London, United Kingdom}{eduard.eiben@rhul.ac.uk}{https://orcid.org/0000-0003-2628-3435}{}
\author{Daniel Lokshtanov}{Department of Computer Science, UC Santa Barbara, United States}{daniello@ucsb.edu}{}{}

\authorrunning{E. Eiben and D. Lokshtanov}

\Copyright{Eduard Eiben and Daniel Lokshtanov}



\category{}

\relatedversion{}

\supplement{}

\acknowledgements{}

\nolinenumbers 

\hideLIPIcs  

\EventEditors{John Q. Open and Joan R. Access}
\EventNoEds{2}
\EventLongTitle{The 36th International Symposium on Computational Geometry (SoCG 2020)}
\EventShortTitle{SoCG 2020}
\EventAcronym{SoCG}
\EventYear{2020}
\EventDate{June 23--26, 2020}
\EventLocation{Z\"{u}rich, Switzerland}
\EventLogo{}
\SeriesVolume{}
\ArticleNo{}
\begin{document}

\maketitle

\begin{abstract}
Given two points in the plane, a set of obstacles defined by closed curves, and an integer $k$, does there exist a path between the two designated points intersecting at most $k$ of the obstacles?
This is a fundamental and well-studied problem arising naturally in computational geometry, graph theory, wireless computing, and motion planning. It remains \textsf{NP}-hard even when the obstacles are very simple geometric shapes ({\em e.g.}, unit-length line segments). 
In this paper, we show that the problem is fixed-parameter tractable (\textsf{FPT}) parameterized by $k$, by giving an algorithm with running time $k^{O(k^3)}n^{O(1)}$. Here $n$ is the number connected areas in the plane drawing of all the obstacles. 
\end{abstract}



\section{Introduction}


In the \gcmor{} problem we are given as input a source point $s$ and a target point $t$ in the plane, and our goal is to move from the source to the target along a continous curve. The catch is that the plane is also littered with obstacles -- each obstacle is represented by a closed curve, and the goal is to get from the source to the target while intersecting as few of the obstacles as possible. Equivalently we can ask for the minimum number of obstacles that have to be removed so that one can move from $s$ to $t$ without touching any of the remaining ones.\footnote{We assume that the regions formed by the obstacles can be computed in polynomial time. We do not assume that the obstacles contain their interiors.  We may assume without loss of generality that the intersection of two obstacles is a 2-D region, if it is not then we can thicken the borders of the obstacles without changing the sets of obstacles they intersect, so that their intersection becomes a 2-D region.}. 
The problem has a wealth of applications, and has been studied under different names, such as {\sc Barrier Coverage} or {\sc Barrier Resilience} in networking and wirless computing~\cite{alt,kirkpatrick3,korman,kumar,kirkpatrick1,yangphd},  or {\sc Minimum Constraint Removal} in planning~\cite{EibenGKY18,lavalle,popov,hauser}.
%
%
%
The problem is NP-hard even when the obstacles are restricted to simple geometric shapes, such as line segments (\eg, see~\cite{alt,kirkpatrick1,yangphd}). On the other hand, for unit-disk obstacles in a restricted setting, the problem can be solved in polynomial time~\cite{kumar}. Whether \gcmor{} can be solved in polynomial time for unit-disk obstacles remains open. The problem is known to be hard to approximate within a factor of $c\log n$ for $c < 1$~\cite{Bandyapadhyay0S18}, and\iflong, perhaps surprisingly,\fi{} no factor $o(n)$-approximation is known.  For restricted inputs (such as unit disc or rectangle obstacles) better approximation algorithms are known~\cite{Bandyapadhyay0S18,kirkpatrick3}. 

In this paper we approach the general \gcmor{} problem from the perspective of parameterized algorithms (see~\cite{CyganFKLMPPS15} for an introduction). In particular it is easy to see that the problem is solvable in time $n^{k+O(1)}$ if the solution curve is to intersect at most $k$ obstacles. Here $n$ is the number of connected regions in the plane defined by the simultaneous drawing of all the obstacles. If $k$ is considered a constant then this is polynomial time, however the exponent of the polynomial grows with the parameter $k$. A natural problem is whether the algorithm can be improved to a {\em Fixed Parameter Tractable} (\FPT{}) one, that is an algorithm with running time $f(k)n^{O(1)}$. In this paper we give the first \FPT{} algorithm for the problem. Our algorithm substantially generalizes previous work by Kumar \etal~\cite{kumar} as well as the first author and Kanj~\cite{EibenK18}.

\begin{theorem}\label{thm:mainGeom}
There is an algorithm for \gcmor{} with running time $k^{O(k^3)}n^{O(1)}$.
\end{theorem}

Our arguments and the relation between our results and previous work are more conveniently stated in terms of an equivalent graph problem, which we now discuss. 
%
Given a graph $G$, a set $C \subset \mathbb{N}$ (interpreted as a set of {\em colors}), and a function $\chi : V(G) \rightarrow 2^C$ that assigns a set of colors to every vertex of $v$, a vertex set $S$ {\em uses} the color set $\bigcup_{v \in S} \chi(v)$. In the \mor{} problem input consists of $G, s, t, \chi$ and $k$, and the goal is to find an $s-t$ path $P$ that uses at most $k$ colors. It is easy to see that \gcmor{} reduces to \mor{} (see Figure~\ref{fig:intro_figure}). 
\begin{figure}[ht]
	\centering
	\includegraphics[width=0.6\textwidth,page=2]{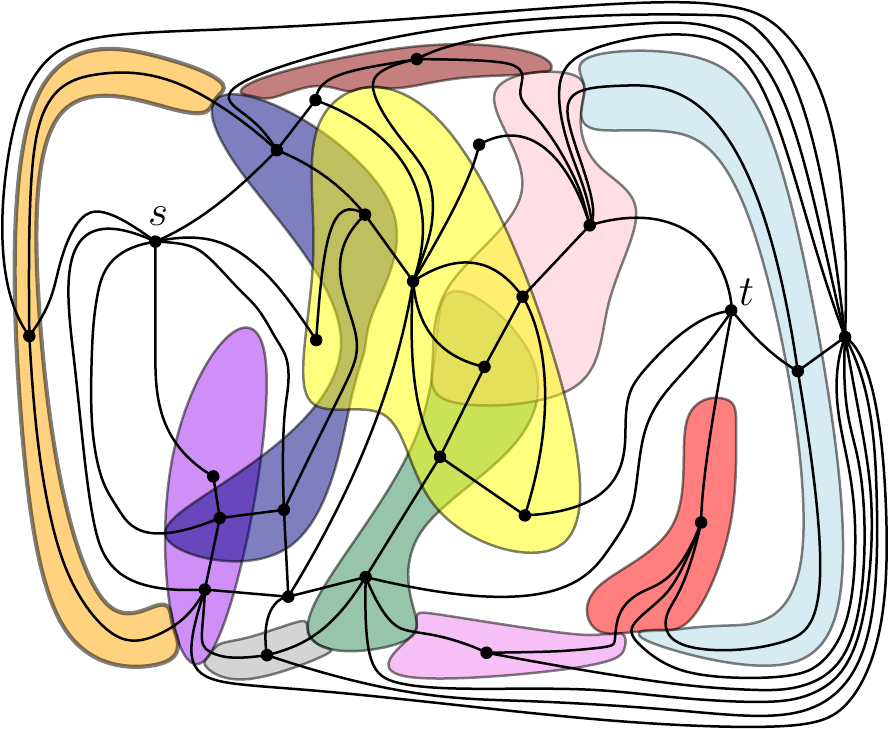}
	\caption{ The figure shows an
		instance of \gcmor{} and the graph $G$ of an equivalent instance of \mor{}. $G$ is the plane graph that is the dual of the plane subdivision determined by the obstacles. Every obstacle corresponds to a color, and the color set of a vertex are the obstacles that contain the vertex in their interior.
	}\label{fig:intro_figure}
\end{figure}
Of course, reducing from  \gcmor{} in this way can not produce all possible instances of \mor: the graph $G$ is always a planar graph, and for every color $c \in C$ the set $\chi^{-1}(c) = \{v \in V(G) : c \in \chi(v)\}$ induces a connected subgraph of $G$. We shall denote the \mor{} problem restricted to instances that satisfy the two properties above by \cmor{}. With these additional restrictions it is easy to reduce back, and therefore \gcmor{} and \cmor{} are, for all practical purposes, different formulations of the same problem. 

\subparagraph*{Related Work in Parameterized Algorithms, and Barriers to Generalization.}
Korman et al.~\cite{korman} initiated the study of \gcmor{} from the perspective of parameterized complexity. They show that \gcmor{} is \FPT{} parameterized by $k$ for unit-disk obstacles, and extended this result to similar-size fat-region obstacles with a constant \emph{overlapping number}, which is the maximum number of obstacles having nonempty intersection. Eiben and Kanj~\cite{EibenK18} generalize the results of Korman et al.~\cite{korman} by giving algorithms for \cmor{} with running time $f(k,t)n^{O(1)}$ and $g(k,\ell)n^{O(1)}$ where $t$ is the treewidth of the input graph $G$, and $\ell$ is an upper bound on the number of {\em vertices} on the shortest solution path $P$. 

Eiben and Kanj~\cite{EibenK18} leave open the existence of an FPT algorithm for \cmor{} - Theorem~\ref{thm:mainGeom} provides such an algorithm. Interestingly, Eiben and Kanj~\cite{EibenK18} also show that if an FPT algorithm for \cmor{} were to exist, then in many ways it would be the best one can hope for. More concretely, for each of the most natural ways to try to generalize Thoerem~\ref{thm:mainGeom}, Eiben and Kanj~\cite{EibenK18} provide evidence of hardness.
Specifically, the \cmor{} problem imposes two constraints on the input -- the graph $G$ has to be planar and the color sets need to be connected. Eiben and Kanj~\cite{EibenK18} show that lifiting {\em either one} of these constraints results in a W[1]-hard problem (i.e. one that is not \FPT{} assuming plausible complexity theoretic hypotheses) {\em even} if the treewidth of the input graph $G$ is a small constant, {\em and} the length of the a solution path (if one exists) is promised to be a function of $k$.

Algorithms that determine the existence of a path can often be adapted to algorithms that find the {\em shortest} such path. Eiben and Kanj~\cite{EibenK18} show that for \cmor{}, {\em this can not be the case!} Indeed, they show that an algorithm with running time $f(k)n^{O(1)}$ that given a graph $G$, color function $\chi$ and integers $k$ and $\ell$ determines whether there exists an $s-t$ path of length at most $\ell$ using at most $k$ colors, would imply that \FPT{} = W[1]. Thus, unless \FPT{} = W[1] the algorithm of Theorem~\ref{thm:mainGeom} can not be adapted to an \FPT{} algorithm that finds a {\em shortest} path through $k$ obstacles.

\subsection{Overview of the Algorithm}
The naive $n^{k+O(1)}$ time algorithm enumerates all choices of a set $S$ on at most $k$ colors in the graph, and then decides in polynomial time whether $S$ 
is a feasible color set, in other words whether there exists a solution path that only uses colors from $S$. At a very high level our algorithm does the same thing, but it only computes sets $S$ that can be obtained as a union of colors of at most $k$ vertices and additionally it performs a pruning step so that not all $n^k$ choices for $S$ are enumerated. 

In \FPT{} algorithms such a pruning step is often done by clever {\em branching}: when choosing the $i$'th vertex defining $S$ one would show that there are only $f(k)$ viable choices 
that could possibly lead to a solution. We are not able to implement a pruning step in this way. Instead, our pruning step is inspired by algorithms based on representative sets~\cite{FominLPS16}.

In particular, our algorithm proceeds in $k$ rounds. In each round we make a family ${\cal P}_i$ of color sets of size at most $i$, with the following properties. First, $|{\cal P}_i| \leq k^{O(k^3)}n^{O(1)}$. Second, if there exists a solution path, then there exists a solution such that the set containing the {\em first $i$ visited} colors is in ${\cal P}_i$.

In each round $i$ the algorithm does two things: first it {\em extends} the already computed families ${\cal P}_{0},\ldots {\cal P}_{i-1}$ by going over every set $S \in \bigcup_{j=0}^{i-1}{\cal P}_{j}$ and every vertex $v \in V(G)$ and inserting $S \cup \chi(v)$ into the new family $\hat{\cal P}_{i}$ if $|S \cup \chi(v)|=i$. It is quite easy to see that $\hat{\cal P}_{i}$ satisfies the second property - however it is a factor of $n$ larger than the union of previous ${\cal P}_{j}$'s. If we keep extending $\hat{\cal P}_{i}$ in this way then after a super-constant number of steps we will break the first requirement that the family size should be at most $k^{O(k^3)}n^{O(1)}$. For this reason the algorithm also performs an {\em irrelevant set} step: as long as $\hat{\cal P}_{i}$ is ``too large'' we show that one can identify a set $S \in \hat{\cal P}_{i}$ that can be removed from $\hat{\cal P}_{i}$ without breaking the first property. We repeat this irrelevant set step until $\hat{\cal P}_{i}$ is sufficiently small. At this point we declare that this is our $i$'th family ${\cal P}_{i}$ and proceed to step $i+1$.

The most technically involved part of our argument is the proof of correctness for the irrelevant set step - this is outlined and then proved formally in Section~\ref{lem:costructingFamily}. This argument crucially exploits the structure of a large set of paths in a planar graph that start and end in the same vertex.

\ifshort
\newcommand{\proofMark}{\spadesuit}

\blfootnote{\noindent {\emph{Statements whose proofs are omitted due to
		space requirements are marked with $(\proofMark)$. 
	}}}
\fi

\section{Preliminaries}
\label{sec:prelim}
For integers $n, m$ with $n \le m$, we let $[n, m] := \{n, n + 1, \ldots , m\}$ and $[n] := [1, n]$.
Let $\mathcal{F}$ be a family of subsets of a universe $U$. A \emph{sunflower} in $\mathcal{F}$ is a subset $\mathcal{F}' \subseteq \mathcal{F}$ such that all pairs of elements in $\mathcal{F}'$ have the same intersection.

\begin{lemma}[\cite{Erdos60,FlumGrohe06}]\label{lem:SF}
	Let $\mathcal{F}$ be a family of subsets of a universe $U$, each of cardinality exactly
	$b$, and let $a \in \mathbb{N}$. If $|\mathcal{F}|\geq b!(a-1)^{b}$, then $\mathcal{F}$
	contains a sunflower $\mathcal{F}'$ of cardinality at least $a$. Moreover,
	$\mathcal{F}'$ can be computed in time polynomial in $|\mathcal{F}|$.
\end{lemma}

We assume familiarity with the basic notations and terminologies in graph theory and parameterized complexity. We refer the reader to the standard books~\cite{CyganFKLMPPS15,DiestelBook,dfbook} for more information on these subjects.

\subparagraph*{{\bf Graphs.}} All graphs in this paper are simple (\ie, loop-less and with no multiple edges).
Let $G$ be an undirected graph. For an edge $e=uv$ in $G$, {\em contracting} $e$ means removing the two vertices $u$ and $v$ from $G$, replacing them with a new vertex $w$, and
for every vertex $y$ in the neighborhood of $v$ or $u$ in $G$, adding an edge $wy$ in the new graph, not allowing multiple edges.
Given a connected vertex-set $S\subseteq V(G)$, {\em contracting} $S$ means contracting the edges between the vertices in $S$ to obtain a single vertex at the end. For a set of edges $E' \subseteq E(G)$, the subgraph of $G$ induced by $E'$ is the graph whose vertex-set is the set of endpoints of the edges in $E'$, and whose edge-set is $E'$.

A graph is {\it planar} if it can be drawn in the plane without edge intersections (except at the endpoints). A {\it plane graph} is a planar graph together with a fixed drawing. Each maximal connected region of the plane minus the drawing is an open set; these are the {\em faces}. \iflong One is unbounded, called the {\em ourter face}.\fi

\iflong Given a graph $G$, a \emph{walk} $W=(v_1,\ldots,v_q)$ in $G$ is a sequence of vertices in $V(G)$ such that for each $i\in \{1,\ldots,q-1 \}$ it holds that $\{v_i,v_{i+1}\}\in E(G)$. A \emph{path} is a walk with all vertices distinct. \fi
Let $W_1=(u_1, \ldots, u_p)$ and $W_2=(v_1, \ldots, v_q)$, $p, q \in \nat$, be two walks such that $u_p =v_1$. Define the \emph{gluing} operation $\circ$ that when applied to $W_1$ and $W_2$ produces that walk
$W_1 \circ W_2 = (u_1, \ldots, u_p, v_2, \ldots, v_q)$. For a path $P=(v_1,\ldots, v_q)$, $q\in \nat$ and $i\in[q]$, we let $\pref{P}{v_i}$ be the \emph{prefix} of the $P$ ending at $v_i$, that is the path $(v_1,v_2,\ldots v_i)$. Similarly, we let $\suff{P}{v_i}$ be the \emph{suffix} of the $P$ starting at $v_i$, that is the path $(v_i,v_{i+1},\ldots v_q)$. 

For a graph $G$ and two vertices $u, v \in V(G)$, we denote by $d_G(u, v)$ the \emph{distance} between $u$ and $v$ in $G$, which is the length (number of edges) of a shortest path between $u$ and $v$ in $G$.

\subparagraph*{{\bf Parameterized Complexity.}}

A {\it parameterized problem} $Q$ is a subset of $\Omega^* \times
\mathbb{N}$, where $\Omega$ is a fixed alphabet. Each instance of the
parameterized problem $Q$ is a pair $(x, k)$, where $k \in \nat$ is called the {\it parameter}. We say that the parameterized problem $Q$ is
{\it fixed-parameter tractable} (\FPT)~\cite{dfbook}, if there is a
(parameterized) algorithm, also called an {\em \FPT-algorithm},  that decides whether an input $(x, k)$
is a member of $Q$ in time $f(k) \cdot |x|^{O(1)}$, where $f$ is a computable function.  Let \FPT{} denote the class of all fixed-parameter
tractable parameterized problems.
By \emph{\FPT-time} we denote time of the form $f(k)\cdot |x|^{O(1)}$, where $f$ is a computable function and $|x|$ is the input instance size.


\subparagraph*{\textbf{\textsc{\mor{}}} and \textbf{\textsc{\cmor{}}}.}%

For a set $S$, we denote by $2^{S}$ the power set of $S$.
Let $G=(V, E)$ be a graph, let $C \subset \mathbb{N}$ be a finite set of colors, and let $\chi: V \longrightarrow 2^{C}$. A vertex $v$ in $V$ is \emph{empty} if
$\chi(v) = \emptyset$. A color $c$ \emph{appears on}, or is \emph{contained in}, a subset $S$ of vertices if $c \in \bigcup_{v \in S} \chi(v)$. For two vertices $u, v \in V(G)$, $\ell \in \mathbb{N}$, a $u$-$v$ walk $W=(u=v_0, \ldots, v_r=v)$ in $G$ is \emph{$\ell$-valid} if $|\bigcup_{i=0}^{r} \chi(v_i)| \leq \ell$; that is, if the total number of colors appearing on the vertices of $W$ is at most $\ell$. A color $c \in C$ is {\em connected} in $G$, or simply {\em connected}, if $\bigcup_{c \in \chi(v)} \{v\}$ induces a connected subgraph of $G$. The graph $G$ is {\em color-connected}, if for every $c \in C$, $c$ is connected in $G$.  

For an instance $(G, C, \chi, s, t, k)$ of \cmor{}, if $s$ and $t$ are nonempty vertices, we can remove their colors and decrement $k$ by $|\chi(s) \cup \chi(t)|$ because their colors appear on every $s$-$t$ path. If afterwards $k$ becomes negative, then there is no $k$-valid $s$-$t$ path in $G$. Moreover, if $s$ and $t$ are adjacent, then the path $(s, t)$ is a path with the minimum number of colors among all $s$-$t$ paths in $G$. Therefore, we will assume:

\begin{assumption}
	\label{ass:sat}
	For an instance $(G, C, \chi, s, t, k)$ of \mor{} or \cmor{}, we can assume that $s$ and $t$ are nonadjacent empty vertices.
\end{assumption}

\begin{definition}\rm
	\label{def:colorcontraction}
	Let $s, t$ be two designated vertices in $G$, and let $x, y$ be two adjacent vertices in $G$ such that $\chi(x) = \chi(y)$. We define the following operation to $x$ and $y$, referred to as a \emph{color contraction} operation, that results in a  graph $G'$, a color function $\chi'$, and two designated vertices $s', t'$ in $G'$, obtained as follows:
	\begin{itemize}
		\item $G'$ is the graph obtained from $G$ by contracting the edge $xy$, which results in a new~vertex~$z$;
		\item $s'=s$ (resp.~$t'=t$) if $s \notin \{x, y\}$ (resp.~$t \notin \{x, y\}$), and $s'=z$ (resp.~$t'=z$) otherwise; 
		\item $\chi': V(G') \longrightarrow 2^{C}$ is 
		defined as $\chi'(w) = \chi(w)$ if $w \neq z$, and $\chi'(z)=\chi(x)=\chi(y)$.
	\end{itemize}
	$G$ is \emph{irreducible} if there does not exist two vertices in $G$ to which the color contraction operation is applicable.
\end{definition}

\begin{observation}\label{obs:colorContraction}
	Let $G$ be a color-connected plane graph, $C$ a color set, $\chi: V \longrightarrow 2^{C}$, $s, t \in V(G)$, and $k \in \mathbb{N}$. Suppose that the color contraction operation is applied to two vertices $x,y$ in $G$ to obtain $G'$, $\chi'$, $s', t'$, as described in Definition~\ref{def:colorcontraction}. For any two vertices $u,v\in V(G)$ and $p\subseteq C$ there is a $u$-$v$ walk $W$ with $\chi(W)=p$ in $G$ if and only if there is a $u'$-$v'$ walk $W'$ with $\chi(W')=p$, where $u'=u$ (resp.~$v'=v$) if $u \notin \{x, y\}$ (resp.~$v \notin \{x, y\}$), and $u'=z$ (resp.~$v'=z$) otherwise.
\end{observation}

\section{FPT algorithm for \cmor}

Given an instance $(G,C, \chi, s,t,k)$ and a vertex $v\in V(G)$, 
we say that a vertex $u$ is \emph{reachable} from a vertex $v$ by a color set $p\subseteq C$ if there exists a $v$-$u$ path $p$ with $\chi(P)\subseteq p$.
Furthermore, we say that a color set $p\subseteq C$ is \emph{$v$-opening}
if there is a vertex $u\in V(G)$ such that $u$ is reachable from $v$ by $p$, but not by any proper subset of $p$. Note that necessarily $\chi(v)\subseteq p$. 
A set of colors $p$ \emph{completes} a $v$-$t$ walk $Q$ if there is an $s$-$v$ path $P$ with $\chi(P)=p$, $|p \cup \chi(Q)|\le k$, and $v$ is the only vertex on $Q$ reachable from $s$ by $p$. We say $p$ \emph{minimally completes} a $v$-$t$ walk $Q$, if $p$ completes $Q$ and there is no $s$-$v$ path $P'$ with $\chi(P')\subsetneq p$. 
We say that an $s$-$t$ path $P$ is \emph{nice}, if for every prefix $\pref{P}{u}$ of $P$ ending at the vertex $u\in V(G)$ there is no $s$-$u$ path $P'$ with $\chi(P')\subsetneq \chi(\pref{P}{u})$. 

\begin{observation}
	There is a $k$-valid $s$-$t$ path if and only if there is a nice $k$-valid $s$-$t$ path.
\end{observation}

\begin{definition}[$k$-representation]\label{def:k-representation_wrt_v}
	Given an instance $(G,C, \chi, s,t,k)$ of \cmor{}, a vertex $v\in V(G)$, and two families $\PPP$ and $\PPP'$ of $s$-opening subsets of $C$ of size $\ell\le k$, we say that $\PPP'$ $k$-represents $\PPP$ w.r.t. $v$ if for every $p\in \PPP$ and every $v$-$t$ 
	walk $Q$ such that $p$ minimally completes $Q$,
	there is a set $p'\in \PPP'$ such that $|p'\cup \chi(Q)|\le k$, $p'\cap \chi(Q)\supseteq p\cap \chi(Q)$, and there is an $s$-$v$ path $P'$ with $\chi(P')= p'$. 
\end{definition}

The main technical result of this paper is then the following theorem stating that if a family $\PPP$ of color sets is large, then we can find an irrelevant color set in $\PPP$.

	\begin{restatable}{lemma}{constructingFamily}\label{lem:costructingFamily}
	Let $(G,C,\chi,s,t,k)$ be an instance of \cmor. Given a family $\PPP$ of $s$-opening color sets of set of size $\ell\le k$ and a vertex $v\in V(G)$, 
	if $|\PPP|> f(k)$, $f(k)= k^{\Oh(k^3)}$, then we can in time polynomial in $|\PPP|+ |V(G)|$ 
	find a set $p\in\PPP$ such that $\PPP\setminus\{p\}$ $k$-represents $\PPP$~w.r.t.~$v$.	
\end{restatable}

\subsection{Algorithm assuming Lemma~\ref{lem:costructingFamily}}

In this subsection, we show how to get an \FPT-algorithm for \cmor\ assuming Lemma~\ref{lem:costructingFamily} is true. The whole algorithm is relatively simple and is given in Algorithm~\ref{alg:main}. The main goal of the subsection is to show that, given Lemma~\ref{lem:costructingFamily}, the algorithm is correct and runs in \FPT-time. 

While the definition of $k$-representation is not the most intuitive definition of representation (for example it is not transitive), we show that it is sufficient to preserve a path of some specific form. 
Let $P$ be a $k$-valid $s$-$t$ path. For $i\in [0,k]$ let $v_i(P)$ be the last vertex on $P$ such that $|\chi(\pref{P}{v_i(P)})|\le i$
and let $\ell_i(P)$ be the length, \ie, number of edges, of $\suff{P}{v_i(P)}$.
If the path $P$ is clear from the context, we write $v_i$ and $\ell_i$ instead of $v_i(P)$ and $\ell_i(P)$. For example, we write $\pref{P}{v_i}$ instead of $\pref{P}{v_i(P)}$. 
Note that for a $k$-valid $s$-$t$ path $P$, $\ell_k(P)=0$ and since $G$ is irreducible w.r.t. color contraction, $\ell_0(P)$ is precisely the length of $P$. For two vectors $(a_0,a_1, a_2,\ldots, a_k), (b_0, b_1, b_2,\ldots, b_k)$ we say $(a_0, \ldots, a_k) < (b_0,\ldots, b_k)$ if there exists $i\in [0,k]$ such that $a_i<b_i$ and for all $j>i$ $a_j=b_j$. 
For a $k$-valid $s$-$t$ path, we call the vector $\vec{\ell}(P)=(\ell_0(P),\ldots, \ell_k(P))$ the \emph{characteristic} vector of $P$ (see also Figure~\ref{fig:charVector}). 
\begin{figure}[ht]
	\centering
	\includegraphics[width=0.8\textwidth,page=5]{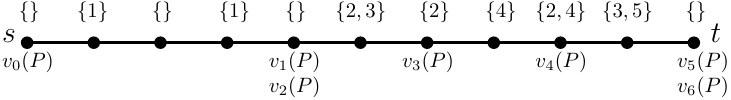}
	\caption{Figure depicting the definition of $v_i(P)$ for $k=6$ and a path using $5$ colors. The characteristic vector $\vec{\ell}(P)=(\ell_0(P),\ldots, \ell_6(P))$ is $(10, 6,6,4,2,0,0)$.}\label{fig:charVector}
\end{figure}

\begin{algorithm}[ht]

	\SetAlgoLined
	\KwData{An instance $(G,C,\chi,s,t,k)$ of \cmor}
	\KwResult{A $k$-valid $s$-$t$ path or \NO, if such a path does not exists}
	
	$\PPP_0=\{\emptyset\}$\label{step:1}; \\
	\For {$i\in [k]$\label{step:2}}{
		$\hat{\PPP}_{i} = \emptyset$\label{step:3} \\
		\For{$v\in V(G)$\label{step:4}}{ 
			\For{$p\in \bigcup_{j\in [0,i-1]}\PPP_j$\label{step:5}}{
				\If{$|\chi(v)\cup p|=i$\label{step:6}}{ 
					\If{there is a $k$-valid $s$-$t$ path $P$ with $\chi(P)\subseteq \chi(v)\cup p$\label{step:7}}{
						Output $P$ and stop\label{step:8}
					}\label{step:9}
					
					$\hat{\PPP}_{i} = \hat{\PPP}_{i}\cup \{\chi(v)\cup p\}$\label{step:10}
				}\label{step:11}
			}\label{step:12}
		}\label{step:13}
		\For{$v\in V(G)$\label{step:14}}{
		$\PPP^v_i=\hat{\PPP}_i$\label{step:15}\\
			\While{$|\PPP^v_i|>f(k)$\label{step:whileLoopBegin}}{
				Compute $p\in \PPP^v_i$ such that $\PPP^v_i\setminus \{p\}$ $k$-represents $\PPP^v_i$ w.r.t. $v$~(by Lemma~\ref{lem:costructingFamily})\label{step:17}\\
				$\PPP^v_i=\PPP^v_i\setminus \{p\}$\label{step:18}
			}\label{step:whileLoopEnd}
		}\label{step:20}
		$\PPP_i=\bigcup_{v\in V(G)}\PPP^v_i$\label{step:21}
	}\label{step:22}
	Output \NO\label{step:23}
	\caption{The algorithm for \cmor{}}	\label{alg:main}
\end{algorithm}

\iflong\begin{lemma}\label{lem:minCharVector}\fi
\ifshort\begin{lemma}[$\proofMark$]\label{lem:minCharVector}\fi
	Let $P$ be a $k$-valid $s$-$t$ path with characteristic vector $\vec{\ell}(P)$, then there exists a nice $k$-valid $s$-$t$ path $P'$ with characteristic vector $\vec{\ell}(P')$ such that $\vec{\ell}(P')\le \vec{\ell}(P)$.
\end{lemma}
\iflong
\begin{proof}
	Let $P'$ be a path such that $\vec{\ell}(P') \le \vec{\ell}(P)$ and there does not exist a path $P''$ with $\vec{\ell}(P'') < \vec{\ell}(P')$. Since $\vec{\ell}(P) \le \vec{\ell}(P)$, the relation $<$ is antisymmetric, and there are at most $n^{k+1}$ different characteristic vectors of a path in an $n$ vertex graph, it follows that such $P'$ always exists. We claim that $P'$ is nice. We prove the claim by contradiction. Assume that $P'$ is not nice and let $v$ be a vertex on $P'$ such that $|\chi(\pref{P'}{v})|=i$, $i\in [k]$, 
	but $v$ can be reached from $s$ by $p\subsetneq \chi(\pref{P'}{v})$.
	Let $P_v$ be an $s$-$v$ path using precisely colors in $p$ and let $P'' = P_v\circ \suff{P'}{v}$. Clearly, $\chi(P'')\subseteq \chi(P')$ and $P''$ is $k$-valid. 
	Moreover, $p=\chi(\pref{P''}{v})\subsetneq \chi(\pref{P'}{v})$ hence $\ell_{|p|}(P'')<\ell_{|p|}(P')$ andl vertices $u\in V(\suff{P'}{v})$, $\chi(\pref{P''}{u})\subseteq \chi(\pref{P'}{u})$ hence $\ell_{j}(P'')\le\ell_{j}(P')$ for all $j\in [|p|,k]$. But then $\vec{\ell}(P'')< \vec{\ell}(P')$, which is a contradiction with the choice of $P'$.
\end{proof}
\fi
The following technical lemma will help us later show that replacing a prefix of a path $P$ with $\chi(\pref{P}{v_i})\in \PPP$ by its representative will always lead to a path $P'$ with $\vec{\ell}(P')\le \vec{\ell}(P)$. 
\newcommand{\pre}{\operatorname{pre}}
\newcommand{\suf}{\operatorname{suf}}
 
\iflong\begin{lemma}~\label{lem:replacingPath}\fi
\ifshort\begin{lemma}[$\proofMark$]~\label{lem:replacingPath}\fi
	Let $P$ be an $s$-$t$ path, $w\in V(P)$, let $\pre=\pref{P}{w}$, $\suf=\suff{P}{w}$, 
	and let $\pre'$ be an $s$-$w$ path such that 
	$|\chi(\pre')\cup (\chi(\pre)\cap \chi(\suf))|\le|\chi(\pre)|$ and $|\chi(\pre')|< |\chi(\pre)|$. Then $\vec{\ell}(\pre'\circ\suf) < \vec{\ell}(P)$.
\end{lemma}
\iflong
\begin{proof}
		Let $|\chi(\pre')|=j$ and let $P'=\pre'\circ \suf$. As $\suff{P}{w}=\suff{P'}{w}=\suf$ and $v_j(P')$ is after $w$ on $P'$, but $v_j(P)$ is before $w$ on $P$, we get $\ell_j(P')<\ell_j(P)$. We now need to show that $\ell_{j'}(P')\le\ell_{j'}(P)$ for all $j'>j$. This is the same as showing that for all $u\in \suf$ it holds that $|\chi(\pref{P'}{u})|\le |\chi(\pref{P}{u})|$.
		
		For $u\in \suf$ let $P_u$ be the subpath of $P$ between $w$ and $u$, that is $\pref{\suf}{u}$. For all $u\in \suf$, we have $\chi(\pref{P}{u})=\chi(\pre)\cup \chi(P_u)$ and $\chi(\pref{P'}{u})=\chi(\pre')\cup \chi(P_u)$. Therefore, we can split the respective sizes of the color sets as follows:
		\begin{gather*}
		|\chi(\pref{P}{u})|=|\chi(\pre)|+|\chi(P_u)\setminus\chi(\pre)|\\
		|\chi(\pref{P'}{u})|=|\chi(\pre)'\cup (\chi(\pre)\cap \chi(P_u)) |+|\chi(P_u)\setminus (\chi(\pre)\cup \chi(\pre'))|.		\end{gather*}
		
		Since $|\chi(\pre')\cup (\chi(\pre)\cap \chi(\suf))|\le|\chi(\pre)|$ and $\chi(P_u)\subseteq \chi(\suf)$, it is easy to see that $|\chi(\pref{P'}{u})|\le |\chi(\pref{P}{u})|$ and the lemma follows. 
\end{proof}
\fi 
Next, we show that $k$-representativity preserve in a sense a representation of a $k$-valid paths with minimal characteristic vector. 
Before we state the next lemma we introduce the following notation. We say that a set of colors $p$ \emph{$i$-captures} a $s$-$t$ path $P$ if $|\chi(\pref{P}{v_i}|=|p|$, $p$ completes $\suff{P}{v_i}$, and $p$ contains $\chi(\pref{P}{v_i}) \cap \chi(\suff{P}{v_i})$.  \iflong The main point of the following two lemmas is to show that if we fix $P$ to be a nice $k$-valid path minimizing $\vec{\ell}(P)$, then our computed representative $\PPP_i$ set will always contain a color set $p$ that $i$-captures $P$. This is useful because for a $k$-valid $s$-$t$ path it holds $\suff{P}{v_k}$ is single vertex path containing $t$. Hence, if $p$ $k$-captures $P$, we obtain that $t$ is reachable from $s$ by $p$.  \fi

\begin{lemma}\label{lem:representativeTransitivity}
	Let $(G,C,\chi,s,t,k)$ be a \YES-instance, $P$ a nice $k$-valid path minimizing $\vec{\ell}(P)$, and $\PPP'$ and $\PPP$ two families of $s$-opening subsets of $C$ of size $i\le k$. If
		$|\chi(\pref{P}{v_i})|=i$, 
		$\PPP'$ $k$-represents $\PPP$ w.r.t. $v_i=v_i(P)$, and 
		there is $p\in \PPP$ such that $p$ $i$-captures $P$.
	Then there is $p'\in \PPP'$ 
	such that $p'$ $i$-captures $P$.
\end{lemma}
\begin{proof}
	Since $|p|=|\pref{P}{v_i}|=i$ and $p$ completes $\suf{P}{v_i}$, it follows from the choice of $P$ and Lemma~\ref{lem:replacingPath} that $p$ minimally completes $P$. 
	Because, $\PPP'$ $k$-represents $\PPP$ w.r.t. $v_i$, it follows that there exists $p'\in \PPP'$ such that $|p'\cup \chi (\suf{P}{v_i})|$, there is a $s$-$v_i$ path $P'$ with $\chi(P')=p'$ and \iflong 
	\[p'\cap \chi(\suff{P}{v_i})\supseteq p\cap \chi(\suff{P}{v_i})\supseteq \chi(\pref{P}{v_i})\cap \chi(\suff{P}{v_i}).\] \fi
	\ifshort 
	$p'\cap \chi(\suff{P}{v_i})\supseteq p\cap \chi(\suff{P}{v_i})\supseteq \chi(\pref{P}{v_i})\cap \chi(\suff{P}{v_i}).$ 
	\fi  
	Where the second containment follows, because $p$ $i$-captures $P$. Therefore $p'$ contains $\chi(\pref{P}{v_i})\cap \chi(\suff{P}{v_i})$. 
	To finish the proof it only remains to show that no vertex on $\suff{P}{v_i}$ other than $v_i$ is reachable from $s$ by $p'$. 
	Assume otherwise and let $w\in V(\suff{P}{v_i})\setminus \{v_i\}$ be the last vertex that is reachable by $p'$. \ifshort We show that $|p'\cup (\chi(\pref{P}{w})\cap\chi(\suff{P}{w}))|\le |\chi(\pref{P}{w})|$ $(\proofMark)$. Since clearly $|p'|=i<|\chi(\pref{P}{w})|$, the lemma then follows by applying Lemma~\ref{lem:replacingPath} and from the choice of $P$.  \fi
	\iflong
	Since $|p'|=i$, it is easy to see that 	
	\begin{gather*} 
	|p'\cup (\chi(\pref{P}{w})\cap\chi(\suff{P}{w}))| = i + |(\chi(\pref{P}{w})\cap\chi(\suff{P}{w}))\setminus p'|.
	\end{gather*} 
	As $p'\cap \chi(\suff{P}{v_i})\supseteq \chi(\pref{P}{v_i}\cap \suff{P}{v_i})$, it holds that everything in $\chi(\pref{P}{v_i}\cap \suff{P}{w})$ is also in $p'$ and it follows that 

\begin{align*}
	|(\chi(\pref{P}{w})\cap\chi(\suff{P}{w}))\setminus p'|&\le |(\chi(\pref{P}{w})\setminus \chi(\pref{P}{v_i}))\cap\chi(\suff{P}{w})|\\
	&\le |\chi(\pref{P}{w})\setminus \chi(\pref{P}{v_i})|\\
	&\le |\chi(\pref{P}{w})|-i
\end{align*}
			
	
	Moreover, $v_i$ is the last vertex on $P$ such that $\pref{P}{v_i}$ uses at most $i$ colors. Hence $|p'|<|\chi(\pref{P}{w})|$ and the lemma follows by applying Lemma~\ref{lem:replacingPath} and from the choice of $P$. 
	\fi
%
	%
\end{proof}

\begin{lemma}\label{lem:preserveSolution}
	Let $(G,C,\chi,s,t,k)$ be a \YES-instance, $P$ a nice $k$-valid $s$-$t$ path minimizing the vector $\vec{\ell}(P)$.
	Moreover, let $\PPP_0=\emptyset$ and $\PPP_1,\ldots, \PPP_k$ the color sets created in the step on line~\ref{step:21} of Algorithm~\ref{alg:main}. 
	Then for all $i\in [0,k]$ such that $|\chi(\pref{P}{v_i})|=i$,
	there is $p_i\in \PPP_i$ 
	such that $p_i$ $i$-captures $P$.
\end{lemma}

\begin{proof}
	We will prove the lemma by induction. 
	Since $\PPP_0$ contains $\emptyset$ and $\chi(s)=\emptyset$, it is easy to see that the lemma is true for $i=0$ and that $\chi(\pref{P}{v_0})=0$. Let us assume that the lemma is true for all $j<i$. If $v_i=v_{i-1}$,\footnote{Throughout the proof, to improve readability we write $v_i$ instead of $v_i(P)$.} then the statement is true for $i$, because $|\chi(\pref{P}{v_i})|\le i-1$. Hence, we assume for the rest of the proof that $v_i\neq v_{i-1}$. 
	Let $j\in [0,i-1]$ be such that $v_{j-1}\neq v_{i-1}$ but $v_{j}= v_{i-1}$ and let $u$ be the vertex on $P$ just after $v_j$. It follows from definition of $v_{j-1}$, $v_j$, and $v_{i-1}$ that $|\chi(\pref{P}{v_j})|=j$ and $|\chi(\pref{P}{u})|=i$.
	By the induction hypothesis there is $p_j\in \PPP_j$ such that $p_j$ $i$-captures $P$. In particular
	$v_j$ is the last vertex on $\suff{P}{v_j}$ reachable from $s$ by $p_j$ and 
	$p_j\supseteq \chi(\pref{P}{v_j})\cap \chi(\suff{P}{v_j}).$
\iflong 
\begin{claim}
\fi
\ifshort 
\begin{claim}[$\proofMark$]
\fi

	$|p_j\cup \chi(u)|=i$ and 
	$p_j\cup \chi(u)$ minimally completes $\suff{P}{v_i}$.
\end{claim}
\iflong
	\begin{cProof}
		First, as $p_j$ completes $\suff{P}{v_j}$, it follows that $|p\cup\chi(u)\cup \chi(\suff{P}{v_i})|\le |p\cup \suf{P}{v_j}|\le k$. 
		
		Second, since $|\chi(\pref{P}{u})|=i=|\chi(\pref{P}{v_i})|$, it follows that $v_i$ is reachable by $\chi(\pref{P}{v_j})\cup \chi (u)$. Moreover, any color $c\in C$ on a vertex on $\suff{P}{v_j}$ between $v_j$ and $v_i$ is either already in $\chi(u)$ or is in $\chi(\pref{P}{v_j})\cap \chi (\suff{P}{v_j})$. Since $v_j$ is reachable by $p_j$ and $p_j\supseteq \chi(\pref{P}{v_j})\cap \chi (\suff{P}{v_j})$,
		 $v_i$ is reachable by $p_j\cup \chi(u)$ from $s$. 
		 
		 Moreover, $|p_j|=|\chi(\pref{P}{v_j})|=j$ and because $p_j\supseteq \chi(\pref{P}{v_j})\cap \chi(\suff{P}{v_j})$ it is not difficult to see that $|p_j\cup \chi(u)|\le |\chi(\pref{P}{v_j})\cup \chi(u)| = i$. 
		If $v_i$ is reachable from $s$ by a 
		a subset (not necessarily proper) $q$ of $p_j\cup \chi(u)$ of size at most $i-1$, then if we replace the prefix $\pref{P}{v_i}$
		by an $s$-$v_i$ path using only colors in $q$, we get, by Lemma~\ref{lem:replacingPath}, a $k$-valid $s$-$t$ path $P'$ with $\ell(P')<\ell(P)$, which is not possible by the choice of $P$ and Lemma~\ref{lem:minCharVector}. Hence $|p_j\cup \chi(u)|=i$.

		Finally, it remains to show that $v_i$ is the only vertex on $\suff{P}{v_i}$ reachable by $p_j\cup \chi(u)$.
		We prove it by contradiction. Let $w\in V(\suff{P}{v_i})\setminus \{v_i\}$ be the last vertex on $P$ that is reachable by $p_j\cup \chi(u)$. Since $p_j\supseteq \chi(\pref{P}{v_j})\cap \chi(\suff{P}{v_j})$, it follows that $(p_j\cup \chi(u)\cup \chi(\pref{\suff{P}{u}}{w}))\supseteq \chi(\pref{P}{w})\cap \chi(\suff{P}{w})$. Moreover, $|\chi(\pref{P}{w})|\ge i+1$ and $|\chi(p\cup\chi(u)|=i$ by the previous claim. Therefore the claim follows by Lemma~\ref{lem:replacingPath}.
	\end{cProof}
\fi
	From the above claim, it follows that $\hat{\PPP}_i$ contains a color set $\hat{p}=p_j\cup\chi(u)$ such that $|\hat{p}|=i$ minimally completes $\suff{P}{v_i}$. 
	Moreover, $\hat{p}\supseteq \chi(\pref{P}{v_i})\cap \chi(\suff{P}{v_i})$ and $\hat{p}$ $i$-captures $P$.
	The rest of the proof follows by applying Lemma~\ref{lem:representativeTransitivity} in every loop between the steps on lines~\ref{step:whileLoopBegin}~and~\ref{step:whileLoopEnd} for $v=v_i$.
\end{proof}

\iflong
Now we are ready to prove the main result of the paper. 
\fi
\ifshort
Now, if the nice $k$-valid $s$-$t$ path $P$ minimizing the vector $\vec{\ell}(P)$ contains $i\le k$ colors, then $v_i(P)$ is a singleton path $(t)$. Since by Lemma~\ref{lem:preserveSolution} there is $p\in \PPP_i$ that $i$-captures $P$, it means that $t$ is reachable from $s$ by $p$ and Algorithm~\ref{alg:main} outputs a $s$-$t$ path using only the colors in $p$. Moreover, whenever it outputs a path it check whether it is $k$-valid. Therefore after analyzing the running time of  Algorithm~\ref{alg:main} we obtain the following theorem.
\begin{theorem}[$\proofMark$]\label{thm:main}
\fi
\iflong
\begin{theorem}\label{thm:main}
	\fi
	
	There is an algorithm that given an instance $(G,C,\chi,s,t,k)$ of \cmor\ either outputs $k$-valid $s$-$t$ path or decides that no such path exists, in time $\Oh(k^{\Oh(k^3)}\cdot |V(G)|^{\Oh(1)})$.
\end{theorem}
\iflong 
\begin{proof}
	Given an instance $(G,C,\chi,s,t,k)$ we simply run Algorithm~\ref{alg:main} and return its output. 
	\begin{claim}
		Algorithm~\ref{alg:main} runs in time $\Oh(k^{\Oh(k^3)}\cdot |V(G)|^{\Oh(1)})$.
	\end{claim}
\begin{cProof}
	Let $n=|V(G)|$.
	The algorithm loops $k$ times and in each loop it goes through all $n$ vertices in $G$ and all at most $k\cdot k^{\Oh(k^3)}\cdot n$ already computed color sets.  For each of $k\cdot k^{\Oh(k^3)}\cdot n^2$ pairs of vertex and color set it first verifies if $|\chi(v)\cup(p)|=i$, if yes it create auxiliary (non-colored) graph $G'$, induced subgraph of $G$, with precisely the vertices $w$ with $\chi(w)\subseteq \chi(v)\cup(p)$ and verify if there is an $s$-$t$ path in $G'$ in time $\Oh(n)$. If such path exists it outputs it and stops. Else it adds $\chi(v)\cup(p)$ to $\hat{\PPP}_i$. 
	It follows that $\hat{\PPP}_i\le k\cdot k^{\Oh(k^3)}\cdot n^2$. Hence, between steps~\ref{step:14}~and~\ref{step:20} Algorithm~\ref{alg:main} runs at most $k\cdot k^{\Oh(k^3)}\cdot n^3$ times the algorithm from Lemma~\ref{lem:costructingFamily}, each of these runs is done in $k^{\Oh(k^3)}\cdot n^{\Oh(1)}$ time. 
\end{cProof}
		\begin{claim}
		Algorithm~\ref{alg:main} correctly solves \cmor.
	\end{claim}
	\begin{cProof}
		Clearly, Algorithm~\ref{alg:main} outputs a path only in step~\ref{step:8} and before it outputs a path it checks whether it is a $k$-valid $s$-$t$ path. Now assume that $(G,C,\chi,s,t,k)$ is a \YES-instance and let $P$ be a nice $k$-valid $s$-$t$ path minimizing the characteristic vector $\vec{\ell}(P)$. Let $i=|\chi(P)|$. Note that $v_i(P)=t$ and $\suff{P}{t}$ is one-vertex path. By Lemma~\ref{lem:preserveSolution} there is $p_i\in \PPP_i$ such that $p_i$ $i$-captures $P$. Therefore,
		$t$ is reachable from $s$ by $p_i$ and hence there is a $k$-valid $s$-$t$ path $P'$ with $\chi(P')\subseteq p_i$. Moreover, as $\PPP_i\subseteq \hat{\PPP}_i$ it follows that $p_i\in \hat{\PPP}_i$ and it would be added to $\hat{\PPP}_i$ in the step on line~\ref{step:10} of Algorithm~\ref{alg:main}. But in the step on line~\ref{step:7} Algorithm~\ref{alg:main} verified whether there is a $k$-valid path $P'$ with $\chi(P')\subseteq p_i$ and then outputted one such path and terminated.
	\end{cProof}
\end{proof}
\fi 
Note that by the reduction from \gcmor\ to \cmor\ discussed in the introduction, Theorem~\ref{thm:main} implies also an algorithm for \gcmor\ with the asymptotically same running time and hence Theorem~\ref{thm:mainGeom}.

\subsection{Proof of Lemma~\ref{lem:costructingFamily}}

\begin{observation}\label{obs:easyCasekRep}
	Let $\PPP$ be a family of $s$-opening subsets of $C$ of size $\ell\le k$, $v\in V(G)$, and $p\in \PPP$. If there is an $s$-$v$ path $P$ with $\chi(P)\subsetneq p$, then $\PPP\setminus \{p\}$ $k$-represents $\PPP$. 
\end{observation}

For the rest of the section we will fix $v\in V(G)$, $\ell\in [k]$, and we let $\PPP$ be a family of $s$-opening color sets of size $\ell$ such that, for every $p\in\PPP$, $v$ is reachable from $s$ by $p$ but is not reachable from $s$ by any proper subset of $p$. Our goal in the remainder of the section is to show that if $|\PPP|> f(k)$, $f(k)=k^{\Oh(k^3)}$, then we can find in \FPT-time a color set $p\in \PPP$ such that $\PPP\setminus \{p\}$ $k$-represents $\PPP$ w.r.t. $v$. We refer to such $p$ also as an \emph{irrelevant} color set.

\subsubsection{Sketch of the Proof}

The main idea is to show that if the family $\PPP$ is large, in our case of size at least $k^{\Oh(k^3)}$, then we can find a subfamily of $\PPP$ that is structured and this structure makes it easier to find an irrelevant color set that can be always represented within the structured subfamily. We can first apply sunflower lemma and restrict our search to a subfamily of size at least $k^{\Oh(k^2)}$ whose color sets pairwise intersect in the same color sets $c$, but are otherwise pairwise color-disjoint. Now we can remove colors in $c$ from the graph and apply the color contraction operation to newly created neighbors with the same color (see Subsection~\ref{subsec:finishingProof}). 

In the rest of the proof, we can restrict our search for an irrelevant color set to a family $\PPP$ whose color sets are pairwise color disjoint. Moreover, we assume the graph is irreducible w.r.t. color contraction. Now for each $p_i\in \PPP$ we compute an $s$-$v$ path $P_i$ such that $\chi(P_i)=p_i$, by Observation~\ref{obs:easyCasekRep} this is simply done by finding an $s$-$v$ path in the subgraph induced on vertices with colors in $p_i$. The goal is to further restrict the search for an irrelevant path to a set of paths $\bfP$ such that there is a small set of vertices $U$, $|U|\le 2k$, such that all the paths in $\bfP$ visit all vertices of $U$ in the same order, but every vertex in $V(G)\setminus (U\cup \{s,v\})$ appears on at most $\frac{|\bfP|}{f(k)}$ paths. This is simply done by finding a vertex that appear on the most paths in $\bfP$, including the vertex in $U$ if the vertex appears on at least $\frac{|\bfP|}{|U|!\cdot f(k)}$ paths, and restricting $\bfP$ to the paths containing the vertex. Otherwise, we stop.
We show in Lemma~\ref{lem:sizeOfU} that because each path in $\bfP$ has at most $k$ colors, we stop after including at most $2k$ vertices into $U$. To get the paths that visit $U$ in the same order, we just go through all $|U|!$ orderings of $U$ and pick the one most paths adhere to. To finish the proof, we show that thanks to the structure of paths in $\bfP$, for any two consecutive vertices in $U$, there is a large set of paths that are pairwise vertex disjoint between the two consecutive vertices of $U$~(Lemma~\ref{lem:nonIntersectingPaths}). Hence, we get into the situation similar to the one in Figure~\ref{fig:intersecting_in_U}. Any $v$-$t$ path (walk) that contains at most $k$ colors and does not contain vertices in $U$ can only interact with a few of these paths between the two consecutive vertices. Hence, because $\PPP$ was large and because of the structure of paths in $\bfP$, we find a path that cannot share a color with any $v$-$t$ walk with at most $k$ colors (Lemma~\ref{lem:main}). But the color set of such a path is then represented by any other color set in $\PPP$, as they have the same size.

\begin{figure}[ht]
	\centering
	\includegraphics[width=0.8\textwidth,page=1]{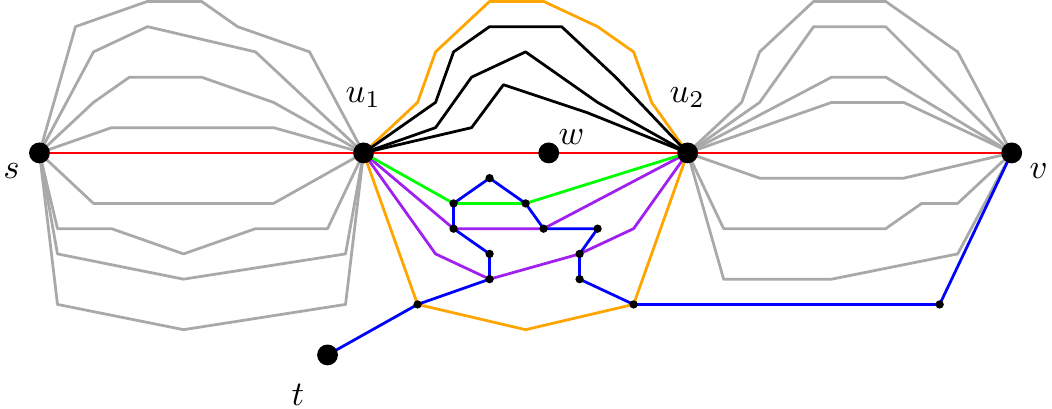}
	\caption{
		A set of pairwise color-disjoint paths that intersects exactly in $u_1$ and $u_2$ in the same order. If a path $P$ from $v$ to $t$ do not contain $s$, $u_1$, nor $u_2$ but it shares a color with some vertex $w$ on the part of the red. Then $P$ has to cross at least $4$ of the color-disjoint path and hence it has to contain at least $3$ colors. 
		For example for the blue path are vertices outside of the orange region, inside the purple region, and the region between red and green path pairwise color-disjoint. In each of these regions the blue path contains at least $2$ consecutive vertices, hence at least~one~is~not~empty.
	}\label{fig:intersecting_in_U}
\end{figure}

\subsubsection{The Color-Disjoint Case}


The goal of this subsection is to show that Lemma~\ref{lem:costructingFamily} is true for a special case when the color sets in $\PPP$ are pairwise color-disjoint and the input graph is irreducible w.r.t. color contraction. This is the most difficult and technical part of the proof. 
For the rest of the subsection we will have the following assumption:

\begin{assumption}\label{ass:color_contracted}
	For an instance $(G,C,\chi,s,t,k)$ of \cmor\ and family $\PPP$ of color sets each of size $\ell\le k$, we assume that $G$ is irreducible w.r.t. color contraction and the sets in $\PPP$ are pairwise color-disjoint.
\end{assumption}

In this subsection, it will be more convenient to work with a set of paths instead of a set of color sets. Given a set $\PPP=\{p_1,\ldots,p_{|\PPP|}\}$ of color-disjoint color sets such that $v$ is reachable by each $p\in \PPP$ from $s$ but not by any proper subset of $p$, we will construct a set of paths $\bfP =\{P_1,\ldots,P_{|\PPP|} \}$ such that $\chi(P_i)=p_i$ for all $i\in[|\PPP|]$. Note that, since $v$ is not reachable from $s$ by any proper subset of $p_i$, this can be simply done by finding a shortest $s$-$v$ path in the graph obtained from $G$ by removing all vertices containing a color not in $p_i$. 

Now we restrict our attention to a subset of paths $\bfQ$ constructed by Algorithm~\ref{alg:refiningPaths}.

\begin{algorithm}[ht]
	\SetAlgoLined
	\KwData{A set of pairwise color-disjoint paths $\bfP$ in a graph $G$}
	\KwResult{A subset $\bfQ$ of $\bfP$ and $U\subseteq V(G)$ such that $|\bfQ| > \frac{|\bfP|}{((|U|+1)!\cdot(8k^2+8k+2))^{|U|}}$, all paths in $\bfQ$ contains all the vertices in $U$, and for every vertex $w\in V(G)\setminus U$ at most $\frac{|\bfQ|}{(|U|+1)!\cdot(8k^2+8k+2)}$ paths in $\bfQ$ contains $w$. }
	$U=\emptyset$ and $\bfQ=\bfP$\label{step:refine1}\\
	let $u$ be a vertex in $V(G)\setminus U$ contained by the highest number of paths in $\bfQ$\label{step:refine2}\\
	\If{$u$ is contained in more than $\frac{|\bfQ|}{(|U|+1)!\cdot(8k^2+8k+3)}$  paths\label{alg2:3}}{
		$U=U\cup \{u\}$\\
		restrict $\bfQ$ to contain only the paths containing $u$\\
		go to the step on line~\ref{step:refine2}\\
	}
	\caption{}\label{alg:refiningPaths}
\end{algorithm}

We will start by showing that when the algorithm is finished, $|U|$ is bounded by $2k$. To show this claim we first need \iflong two topological lemmas\fi\ifshort the following topological lemma\fi.
\iflong 
\begin{lemma}[Lemma 4.8 in the full version of \cite{EibenK18}]
	\label{lem:generalpaths}
	Let $G'$ be a plane graph, and let $x, y,z \in V(G')$. Let $x_1, \ldots, x_r$, $r \geq 3$, be the neighbors of $x$ in counterclockwise order. Suppose that, for each $i \in [r]$, there exists an $x$-$y$ path $P_i$ containing $x_i$ such that $P_i$ does not contain $z$ and does not contain any $x_j$, $j \in [r]$ and $j \neq i$. Then there exist two paths $P_i, P_j$, $i, j \in [r]$ and $i \neq j$, such that the two paths $P_i, P_j$ induce a Jordan curve separating $\{x_1, \ldots, x_r\} \setminus \{x_i, x_j\}$ from $z$.
\end{lemma}
\fi
\iflong 
\begin{lemma}\label{lem:repetingSameColorSet}
\fi
\ifshort 
\begin{lemma}[$\proofMark$]\label{lem:repetingSameColorSet}
	\fi
	Let $G$ a color-connected plane graph that is irreducible w.r.t. color contraction, $s,u_1,u_2, u_3,v$ be vertices in $G$ and let $\bfP = \{P_1,\ldots , P_{|\bfP|}\}$ be pairwise color-disjoint $s$-$v$ paths all going through the vertices $u_1$, $u_2$, and $u_3$ in the same order. Then there are at most two paths $P_i\in \PPP$ such that if $w_j^i$, $j\in [3]$, denotes the vertex on $P_i$ immediately after $u_j$ then $\chi(w_1^i)\cap \chi(w_3^i)\neq \emptyset$.
\end{lemma}
\iflong 
\begin{proof}
	Since the paths in $\bfP$ are color-disjoint, it follows that the vertices $s,u_1,u_2, u_3,v$ are empty. Moreover, $G$ is irreducible w.r.t. color contraction. Therefore, all $w_j^i$'s are not empty and $w_j^i$ and $w_{j}^{i'}$ are different vertices whenever $i\neq i'$.
	Applying Lemma~\ref{lem:generalpaths} to $G$, vertices $u_1, u_2, u_3$, and the restriction of the paths to the subpaths between $u_1$ and $u_2$. We get that there are two paths $P_{j}$,$P_{j'}$, $j,j'\in [|\bfP|]$ that induce a Jordan curve separating $w_1^i$'s, for all paths $P_i$, $i\in [|\bfP|]\setminus\{j,j'\}$, from $u_3$. But $w_3^i$ is a neighbor of $u_3$. Moreover $w_3^i$ is not empty, therefore it cannot appear on $P_j$ nor $P_{j'}$.  
	Hence, the same Jordan curve separates $w_1^i$ and $w_3^i$. Since the paths are color-disjoint, this Jordan curve does not contain any color on $P_i$. Since $G$ is color-connected, we get that $\chi(w_1^i)\cap \chi(w_3^i)=\emptyset$.
\end{proof}
\fi
Now we can show that if $|U|\ge 2k+1$, then at the point when Algorithm~\ref{alg:refiningPaths} adds $2k+1$-st element to $U$, we can find $k^2+k+1$ paths in $\bfQ$ that visit the first $2k+1$ vertices of $U$ in the same order. Lemma~\ref{lem:repetingSameColorSet} then implies that there is a path $P_i\in \bfP$ such that $\chi(w_j^i)\cap \chi(w_{j'}^i) = \emptyset$ for all $j\neq j'$, $j,j'\in \{1,3,5, \ldots, 2k+1 \}$, where $w_j^i$ denotes the vertex on $P_i$ immediately after $u_j$. Then $|\chi(P_i)|\ge k+1$ which contradicts definition of $\bfP$. 
\iflong
\begin{lemma}\label{lem:sizeOfU}
\fi
\ifshort
\begin{lemma}[$\proofMark$]\label{lem:sizeOfU}
	\fi
	If $|\bfP|\ge f(k)$, $f(k)=k^{\Oh(k^2)}$, 
	then when Algorithm~\ref{alg:refiningPaths} terminates, it holds that $|U|< 2k+1$.
\end{lemma}
\iflong 
\begin{proof}
We show that the lemma holds for $f(k) = ((2k+1)!\cdot(8k^2+8k+3))^{2k+1}\cdot (k^2+k)\cdot (2k+1)!+1$, which is easily seen to be in $k^{\Oh(k^2)}$.
Assume this is not the case and $|U| \ge 2k+1$. Let $U'$ be the first $2k+1$ vertices of $U$ found by the previous algorithm and let $\bfQ'$ be the subset of the paths in $\bfP$ that contains all vertices in $U'$. Clearly, there are least $\lceil\frac{|\bfP|}{((2k+1)!\cdot(8k^2+8k+3))^{2k+1}}\rceil\ge (k^2+k)\cdot (2k+1)!+1$ paths in $\bfQ'$ and hence there is an ordering of $U'$ such that at least $k^2+k+1$ paths visit vertices of $U'$ in this order, let $\bfQ''$ be the restriction of $\bfQ'$ to these paths. 
Let $\bfQ''=\{P_1,\ldots,P_{|\bfQ''|} \}$ and for $i\in [|\bfQ''|]$ and $j\in [2k+1]$ let $w_j^i$ be the vertex immediately after $u_j$ on $P_i$. Since the path in $\bfQ''$ are color-disjoint, all the vertices in $U$ are empty. Moreover, $G$ is color contracted, hence $\chi(w_j^i)\neq \emptyset$. By Lemma~\ref{lem:repetingSameColorSet}, $\chi(w_j^i)\cap \chi(w_{j'}^i)\neq \emptyset$ for $|j-j'|\ge 2$ for at most $2$ paths. Therefore, if we have more than $2\cdot \binom{k+1}{2}=k^2+k$ paths in $\bfQ''$, then there is a path such that $\chi(w_j^i)\cap \chi(w_{j'}^i) = \emptyset$ for all $j\neq j'$, $j,j'\in \{1,3,5, \ldots, 2k+1 \}$. But $|\chi(P_i)|\ge |\chi(w_1^i)\cup \chi(w_3^i)\cup \ldots \cup \chi(w_{2k+1}^i)|$. Since the sets $\chi(w_j^i)$, $j\in \{1,3,5, \ldots, 2k+1 \}$, are pairwise color-disjoint and non-empty, we get $|\chi(P_i)|\ge k+1$. But $P_i$ is a $k$-valid path, contradiction.
\end{proof}
\fi

Now we have bounded $|U|$ and the number of paths intersecting in any vertex outside $U$. We first fix an ordering $\tau = (u_1,u_2,\ldots, u_{|U|})$ of vertices in $U$ which maximizes the number of paths in $\bfQ$ that visit $U$ in the same order as $\tau$ and let $\bfQ'$ be the restriction of $\bfQ$ to the paths that are consistent with this ordering. Clearly $|\bfQ|\le |\bfQ'|\cdot (2k)!$ and it suffice to show that we can find an irrelevant path in $\bfQ'$ if $|\bfQ'|$ is large. The agenda for the rest of the proof is as follows. Because $|U|\le 2k$ and intersection number of each vertex outside $|U|$ is small compared to the size of $\bfQ'$, only "few" paths can share a color with any $k$-valid $v$-$t$ walk that do not contain a vertex in $U$ hence we can find an irrelevant path. The color set of this irrelevant path is then the irrelevant color set in $\PPP$.

\iflong 
Let us first show the following simple setting, where the paths in $\bfQ'$ intersects pairwise precisely in the vertices of $U$. While this lemma is not necessary for our proof, it gives an intuition what kind of a structure/arguments we are looking for if the intersection outside of $U$ is small.

\begin{lemma}\label{lem:intesecting_in_U}
	Let $\bfQ'$ be a set of $k$-valid color-disjoint $s$-$v$ paths that pairwise intersects precisely in vertices $u_1, \ldots, u_r$, $r\le k$, in the same order. If $|\bfQ'|>4k\cdot (r+1)$, then we can in polynomial time find a path $P\in \bfQ'$ such that $\chi(P)\cap \chi(Q)=\emptyset$ for every $k$-valid $v$-$t$ walk $Q$ that do not contain any vertex in $U\cup\{s,v\}$ as inner vertex.  
\end{lemma} 

\begin{proof}
	See also Figure~\ref{fig:intersecting_in_U}.
	Let us first restrict our attention to the restriction of the paths between two consecutive vertices in $U\cup \{s,v\}$. Let us for convenience denote $s$ by $u_0$ and $v$ by $u_{|U|+1}$ and let these two vertices be $u_i$ and $u_{i+1}$ and let us denote $P_j^i$ the restriction of $P_j$ to the subpath between $u_i$ and $u_{i+1}$. The paths between $u_i$ and $u_{i+1}$ pairwise only intersect in $u_i$ and $u_{i+1}$. Let $H$ be the plane subgraph of $G$ induced by restriction of paths in $\bfQ'$ to subpaths between $u_i$ and $u_{i+1}$. Let us assume that $P^i_1,\ldots, P^i_{|\bfQ'|}$ are ordered in counterclockwise order around $u_i$ such that $t$ is in the face of $H$ bounded by $P^i_1$ and $P^i_{|\bfQ'|}$. Now let $j\in [|\bfQ'|]$ be such that $2k+1\le j\le |\bfQ'|-2k$. The union of $P^i_{j-1}$ and $P^i_{j+1}$ forms a vertex separator between $t$ and $P^i_j$. Moreover, $G$ is color-connected and paths in $\bfQ'$ are pairwise color-disjoint. Therefore, any $v$-$t$ walk $Q$ that contains a color of $P_j^i$ has to contain a vertex $w$ inside the region bounded by $P^i_{j-1}$ and $P^i_{j+1}$. Now, let us restrict our attention to a $w$-$t$ path $Q'$ that is contained in $Q$. Since $Q$ does not contain $u_i$ nor $u_{i+1}$ as inner vertex the path $Q'$ has to either cross all paths in $\bfP_1=\{P^i_1,P^i_2,\ldots, P^i_{j-1}\}$, or all the paths in $\bfP_2=\{P^{i}_{j+1},P^i_{j+2},\ldots, P^i_{|\bfQ'|}\}$. Let us assume without loss of generality that $Q'$ cross all the paths in $\bfP_1$. Now consider following $k+1$ faces in $H$: $f_1$ bounded by $P_{1}$ and $P_{|\bfQ'|}$, $f_2$ bounded by $P_2$ and $P_3$,$\ldots$, $f_{i'}$ bounded by $P_{2i'-2}$ and $P_{2i'-1}$, $\ldots$, and $f_{k+1}$ bounded by $P_{2k}$ and $P_{2k+1}$. Since $j\ge 2k+1$ and $Q'$ crosses all the paths in $\bfP_1$, $Q'$ has to contain at least two consecutive vertices that are either on the boundary or on the interior of each $f_{i'}$ for $i'\in [k+1]$. As $G$ is color contracted, at least one of two neighbors is always non-empty. Let $w_{i'}$ be a colored vertex in $f_{i'}$. Moreover, for $j'\neq i'$ the boundaries of $f_{i'}$ and $f_{j'}$ are color-disjoint. Therefore, $\chi(w_{i'})\cap \chi(w_{j'})=\emptyset$.
	It follows that $|\chi(Q')|\ge |\bigcup_{i'\in [k+1]}\chi(w_{i'})|\ge k+1$. However, $Q'$ is a path containing only vertices in $Q$, hence also $|\chi(Q)|\ge k+1$, contradiction with the choice of $Q$. Hence, $\chi(P^i_j)\cap\chi(Q)=\emptyset$. It follows that at most $4k$ paths can share a color with any $v$-$t$ walk with at most $k$ colors between $u_i$ and $u_{i+1}$ for $i\in [0,|U|]$. Hence, there are at most $4k\cdot (|U|+1)$ many paths that can share a color with any $k$-valid $v$-$t$ walk and we can find them easily by marking $4k$ paths closest to $t$ between each $u_i$ and $u_{i+1}$. 
\end{proof}
\fi 

Recall that due to Assumption~\ref{ass:color_contracted}, we assume that the graph $G$ is color contracted and no two neighbors have the same color set. Moreover, the paths in $\bfQ'$ are color-disjoint, so the vertices in $U\cup \{s,v\}$ are all empty and each neighbor of these vertices belongs to at most one path in $\bfQ'$. The goal in the following few technical lemmas is to show that for any two consecutive vertices $u_i$ and $u_{i+1}$ in $U$ we can find a large (of size at least $4k+1$) subsets of paths in $\bfQ'$ that pairwise do not intersect between $u_i$~and~$u_{i+1}$.

\begin{figure}[ht]
	\centering
	\includegraphics[width=0.8\textwidth,page=2]{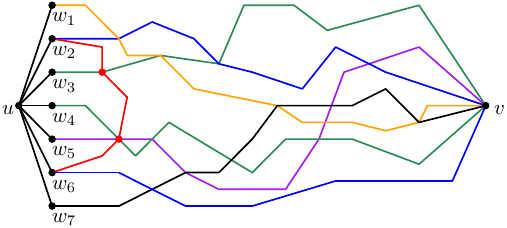}
	\caption{Situation in Lemma~\ref{lem:disjointPaths}. On the picture are seven $u$-$v$ paths, no $3$ of them intersecting in the same vertex. The red $w_2$-$w_6$ path on the picture intersects the three paths containing $w_3$, $w_4$, and $w_5$, respectively. Any such path has to contain at least $2$ vertices, else the only vertex on the path would be the intersection of $3$ $u$-$v$ paths. }\label{fig:smallIntersection}
\end{figure}

\begin{lemma}\label{lem:disjointPaths}
	Given an instance $(G, C, \chi, s, t, k)$ which is irreducible w.r.t. color contraction, two vertices $u$, $v$, $b\in \nat$ and a set $\bfP$ of $k$-valid $u$-$v$ paths such that no $b$ paths intersect in the same vertex. Let $w_1,\ldots, w_r$ be the neighbors of $u$, each the second vertex of a different path in $\bfP$, in counterclockwise order. For $i\in [r]$ let $P_i$ denote the path in $\bfP$ containing $w_i$.  Let $1\le i<j\le r$, then the shortest curve $\sigma$ from $w_i$ to $w_j$ that intersects $G$ only in vertices of $V(G)\setminus \{u,v\}$ contains at least $\frac{\min\{j-i, r+i-j\}-1}{b}$ vertices on paths in $\PPP\setminus \{P_i, P_j\}$.
\end{lemma}

\begin{proof}
	See an example of the situation in Figure~\ref{fig:smallIntersection}. Given a curve $\sigma$, we can easily find a closed curve $\sigma'$ that intersect $G$ in $u$, $w_i$, $w_j$ and the vertices that are intersected by $\sigma$. 
	The vertices on $\sigma'$ are then the vertex separator separating $v$ from either $w_{i+1},\ldots, w_{j-1}$ or from $w_1,\ldots, w_{i-1}$ and $w_{j+1},\ldots, w_r$. If the vertices on $\sigma'$ are the vertex separator separating $v$ from $w_{i+1},\ldots, w_{j-1}$, then all the paths $P_{i+1},\ldots, P_{j-1}$ has to pass a vertex on $\sigma$ different than $w_i$ or $w_j$. Since no $b$ paths intersect in the same vertex, we get that $\sigma$ contains at least $\frac{j-i-1}{b}$ vertices in this case. The case when the vertices on $\sigma'$ are the vertex separator separating $v$ from $w_1,\ldots, w_{i-1}$ and $w_{j+1},\ldots, w_r$ is symmetric and the lemma follows.
\end{proof}
\iflong 
\begin{lemma}\label{lem:disjointPathFewFaces}
	\fi
\ifshort 
\begin{lemma}[$\proofMark$]\label{lem:disjointPathFewFaces}
\fi
	Let $(G, C, \chi, s, t, k)$ be an instance of \cmor{} such that $G$ is irreducible w.r.t. color contraction, $H$ a subgraph of $G$, and $P$ a $k$-valid $u$-$v$ path with $u,v\in V(H)$ and $\chi(P)\cap \chi(H) = \emptyset$. Then $P$ intersects at most $k$ faces of $H$. 
\end{lemma}
\iflong 
\begin{proof}
	Since $P$ is color-disjoint from $H$, $P$ intersects $H$ only in empty vertices. Moreover, because $G$ is irreducible w.r.t. color contraction, it follows that $P$ does not contain two consecutive empty vertices and hence $P$ contains a colored vertex in every face it intersects. Finally, the vertices incident to a face in $H$ form a separator between the vertices of $G$ that lie inside and the vertices of $G$ that lie outside of the face. Since $G$ is color-connected, any color that appear inside two distinct faces of $H$ appears also on a vertex of $H$. Finally, $P$ contains at most $k$ colors and in each face of $H$ it intersects it has at least one color that is unique to this face. Therefore, $P$ intersects at most $k$ faces of $H$.  
\end{proof}
\fi 
The combination of the two above lemma immediately yields the following:
\iflong 
\begin{lemma}\label{lem:nonIntersectingPaths}
\fi
\ifshort 
\begin{lemma}[$\proofMark$]\label{lem:nonIntersectingPaths}
	\fi
	Given an instance $(G, C, \chi, s, t, k)$ which is irreducible w.r.t. color contraction, two vertices $u$, $v$, an integer $b\in \nat$ and a set $\bfP$ of $k$-valid pairwise color-disjoint $u$-$v$ paths such that no $b$ paths intersect in the same vertex. Let $w_1,\ldots, w_r$ be the neighbors of $u$, each the second vertex of a different path in $\PPP$, in counterclockwise order. Let $1\le i<j\le r$ and let $P_i$ and $P_j$ be the two paths in $\PPP$ containing $w_i$ and $w_j$, respectively. If $\min\{j-i, r+i-j\}> 2k\cdot b$, then $P_i$ and $P_j$ do not intersect.
\end{lemma}
\iflong 
\begin{proof}
	Let $\bfP'=\bfP\setminus \{P_i, P_j\}$. By Lemma~\ref{lem:disjointPaths} the shortest curve $\sigma$ from $w_{i-1}$ to $w_j$ that intersects $G$ only in vertices of $V(G)\setminus \{u,v\}$ contains at least $2k$ vertices on paths in $\bfP'$. Let $H$ be the subgraph of $H$ induced by paths in $\bfP'$. By Lemma~\ref{lem:disjointPathFewFaces} both $P_i$ and $P_j$ intersect at most $k$ faces of $H$. If $P_i$ and $P_j$ intersects, then these $2k$ faces form one connected component and there is a curve from $w_i$ to $w_j$ that intersects at most $2k-1$ vertices of $H$, which are precisely the vertices on paths in $\bfP'$, a contradiction.
\end{proof}
\fi

\begin{lemma}\label{lem:main}
	If no $b$ paths in $\bfQ'$ intersect in the same vertex in $V(G)\setminus (U\cup \{s,v\})$ and $|\bfQ'| > (8k^2+8k+2)\cdot (|U|+1)\cdot b$, then we can in polynomial time find a path $P\in \bfQ'$ such that for every $k$-valid $v$-$t$ walk $Q$ that does not contain a vertex in $U$ holds $\chi(P)\cap \chi(Q) = \emptyset$. 
\end{lemma}

\begin{proof}
	For the convenience let us denote $s$ by $u_0$ and $v$ by $u_{|U|+1}$. 
	We will show that for every $i\in \{0,\ldots, |U|\}$, every $k$-valid $v$-$t$ walk can intersect at most $(8k^2+8k+2)\cdot b$ paths in a vertex on the path between $u_i$ and $u_{i+1}$. For a path $P\in \bfQ'$ let $P^i$ denote the subpath between $u_i$ and $u_{i+1}$ and let $\bfQ^i = \{P^i\mid P\in \bfQ'\}$. Clearly, the paths in $\bfQ^i$ are color-disjoint $u_i$-$u_{i+1}$ each containing at most $\ell\le k$ colors and no $b$ paths in $\bfQ^i$ intersect in the same vertex beside $u_i$ and $u_{i+1}$. Now let $H^i$ be the subgraph of $G$ induced by the edges on paths in $\bfQ^i$. Since $G$ is color contracted, $u_i$ is an empty vertex, and the paths in $\bfQ^i$ are colored disjoint, each neighbor of $u_i$ appears on a unique path in $\bfQ^i$. Let $w_1,w_2,\ldots, w_{|\bfQ^i|}$ be the neighbors of $u_i$ in $H^i$ in counterclockwise order and let $P_j^i$ be the path in $\bfQ^i$ that contains $w_j$. Clearly, $t$ is in the interior of some face $f$ of $H^i$ and there is at least one path that contains an edge incident on $f$ in $H^i$. Without loss of generality let $P_1^i$ be such path (note that we can always choose a counterclockwise order around $u_i$ for which this is true). 
	
	\begin{figure}[ht]
		\centering
		\includegraphics[width=0.8\textwidth,page=4]{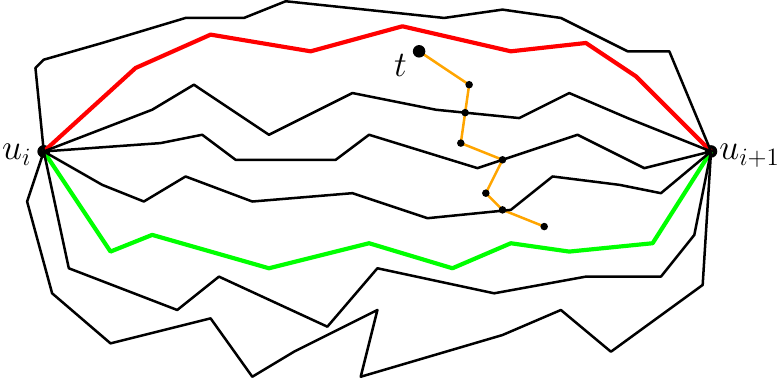}
		\caption{ Any path that starts in a face incident on the red path and finish in a face incident on the green path that does not contain $u_i$ nor $u_{i+1}$ has appear in at least $4$ different faces. Since the paths are color-disjoint, only the consecutive faces can share colors and hence any such path contains at least $2$ colors. }\label{fig:claim}
	\end{figure}

	\begin{claim}
		Let $j\in [|\bfQ^i|]$.  If $ (2k+1)(2k+1)\cdot b <j < |\bfQ^i| -(2k+1)(2k+1)\cdot b$,  $k$-valid $v$-$t$ walk $Q$ that does not contain $u_i$ nor $u_{i+1}$ in the interior holds $\chi(P^i_j)\cap \chi(Q) = \emptyset$.		
	\end{claim}
	\begin{cProof}
		Consider the following set of paths: $P_1^i, P_{2k+2}^i, P_{4k+3}^i, \ldots, P_{4k^2+4k+1}, P_j^i$, $P_{j+2k+1}^i$, $P_{j+4k+2}^i$, $\ldots, P_{j+4k^2+4k}^i$. By Lemma~\ref{lem:nonIntersectingPaths}, these paths are pairwise non-intersecting. Hence, we are in the situation as depicted in Figure~\ref{fig:claim}. Since the paths in $\bfQ^i$ are pairwise color-disjoint, the colors of $P_j^i$ are only on vertices of $G$ inside the region bounded by $P_{2k^2+k+1}$ and $P_{j+2k+1}^i$. Therefore, if $\chi(Q)\cap P_j^i\neq \emptyset$ for some $v$-$t$ walk $Q$, then $Q$ contains a vertex $w$ inside the region bounded by $P_{2k^2+k+1}$ and $P_{j+2k+1}^i$. Moreover, $Q$ does not contain $u_i$ nor $u_{i+1}$ as an inner vertex then it either crosses all the paths in $\bfP_1=\{P_{2k+2}^i, P_{4k+3}^i, \ldots, P_{4k^2+4k+1}\}$ or all the paths in $\bfP_2=\{P_{j+2k+1}^i, P_{j+4k+2}^i,\ldots, P_{j+4k^2+2k}^i\}$. Without loss of generality, let us assume that $Q$ crosses all the paths in $\bfP_1$. The other case is symmetric. As $G$ is color contracted, no two consecutive vertices of $P$ are empty.
		Hence, $Q$ either crosses a path in $\bfP_1$ in a colored vertex or there is a colored vertex on $Q$ between two consecutive paths in $\bfP_1$ (resp. $\bfP_2$). Let us partition the paths in $\bfP_1\cup \{P_1,P_j\}$ into $k+1$ group of two consecutive pairs. that is we partition $\bfP_1$ into groups $\{P_1, P_{2k+2}\}$, $\{P_{4k+3},P_{6k+4}\}$,$\ldots$,$\{P_{4k^2-1},P_{4k^2+2k}\}$,$\{P_{4k^2+4k+1},P_{j}\}$. If the walk $Q$ crosses all paths in $\bfP_1$, it has to contains a colored vertex in each of the $k+1$ groups. However, each two groups are separated by color-disjoint paths. Therefore, two colored vertices in two different groups have to be color-disjoint. But then $\chi(Q)$ contains at least $k+1$ colors, this is however not possible, because $Q$ is $k$-valid. 		
	\end{cProof}
	The lemma then straightforwardly follows from the above claim by marking for each of $|U|+1$ consecutive pairs $2(2k+1)^2\cdot b$ paths that can share a color with some $Q$ and outputting any non-marked path.
\end{proof}

Since $\chi(P)\cap \chi(Q)=\emptyset$, $\chi(P)$ can be replaced by any other color set of $|\chi(P)|$ colors and we can safely remove it from $\PPP$. Since we chose $\bfQ'$ such that no $\frac{|\bfQ|}{(|U|+1)!\cdot(8k^2+8k+3)}=\frac{|\bfQ'|}{(|U|+1)\cdot(8k^2+8k+3)}$ paths intersect in $\bfQ'$, we get the following main result of this subsection.
\iflong  
\begin{lemma}\label{lem:costructingDisjointFamily}
\fi 
\ifshort  
\begin{lemma}\label{lem:costructingDisjointFamily}
\fi 
	Let $(G,C,\chi,s,t,k)$ be an instance of \cmor\ such that $G$ is irreducible w.r.t. color contraction. Given a family $\PPP$ of pairwise color-disjoint $s$-reachable color sets of set of size $\ell\le k$ and a vertex $v\in V(G)$, 
	if $|\PPP|> 2^{\Oh(k^2\log(k))}$, then we can in time polynomial in $|\PPP|+|V(G)|$
	find a set $p\in\PPP$ such that $\PPP\setminus\{p\}$ $k$-represents $\PPP$ w.r.t. $v$.	
\end{lemma}
\iflong 
\begin{proof}
	We start by finding for each $p_i\in \PPP$ an $s$-$v$ path $P_i$ in the graph induced on the vertices $w$ with $\chi(w)\subseteq p_i$. This step can be implemented on a planar graph in $\Oh(|V(G)|)$ time. If $\chi(P_i)\subsetneq p_i$, it follows from Observation~\ref{obs:easyCasekRep} that $\PPP\setminus p_i$ $k$-represents $\PPP$. Hence, for all $p_i\in \PPP$ it holds $\chi(P_i)=p_i$. Now we invoke Algorithm~\ref{alg:refiningPaths} to find a subset of these paths $\bfQ$ and a set of vertices $U$ such that $|U|\le 2k$ (Lemma~\ref{lem:sizeOfU}) and $|\bfQ|> \frac{|\PPP|}{((|U|+1)!\cdot (8k^2+8k+3))^{|U|}}$, and each vertex in $V(G)\setminus (U\cap \{s,v\})$ appears on at most $\frac{|\bfQ|}{(|U|+1)!(8k^2+8k+3)}$. Each of at most $2k$ loops of Algorithm~\ref{alg:refiningPaths} can be implemented in time $|\PPP|\cdot |V(G)|$. Afterwards, we select a subset $\bfQ'$ of $\bfQ$ of paths that visits vertices in $U$ in the same order of the maximum size. This is done by going through each path in $\bfQ$ once and assigning it to the subset with the same order of vertices in $U$ and then selecting the largest subset. Clearly, $\bfQ' \ge \frac{|\bfQ|}{|U|!}$ and therefore each vertex $V(G)\setminus (U\cap \{s,v\})$ appears on at most $b=\frac{|\bfQ'|}{(|U|+1)(8k^2+8k+3)}$ paths in $\bfQ'$. Therefore $|\bfQ'|>(8k^2+8k+2)\cdot (|U|+1)\cdot b$ and we can, by Lemma~\ref{lem:main}, in polynomial time find a path $P_i\in \bfQ'$ such that for every $k$-valid $v$-$t$ walk that does not contain a vertex in $U$ holds $\chi(P_i)\cap \chi(Q)=\emptyset$. Since vertices in $U$ are on $P_i$, for every $v$-$t$ walk $Q$ such that $|\chi(P_i)\cup\chi(Q)|\le k$ and $v$ is the only vertex on $Q$ reachable form $s$ by $\chi(P_i)$ it holds that $\chi(P_i)\cap\chi(Q)=\emptyset$. Since all sets in $\PPP$ have the same size, it holds for every $p'\in \PPP\setminus \{\chi(P_i) \}$ that $|p'\cup \chi(Q)|\le k$ and $p'\cap\chi(Q)\supseteq \chi(P_i)\cap \chi(Q)$. Therefore $\PPP\setminus \{\chi(P_i)\}$ $k$-represents $\PPP$.
\end{proof}
\fi

\subsubsection{Finishing the Proof}\label{subsec:finishingProof}


Given Lemma~\ref{lem:costructingDisjointFamily}, we are ready to proof Lemma~\ref{lem:costructingFamily}. 


\iflong 
\constructingFamily*
\begin{proof}
\fi 
\ifshort
\begin{proof}[Proof of Lemma~\ref{lem:costructingFamily}]
\fi
	Since each set in $\PPP$ has precisely $\ell\le k$ colors, if $|\PPP|>\ell!\cdot (g(k))^{\ell+1}$, $g(k)=k^{\Oh(k^2)}$ then, by Lemma~\ref{lem:SF} we can, in time polynomial in $|\PPP|$, find a set $\QQQ$ of $g(k)+1$ sets in $\PPP$ such that there is a color set $c\subseteq C$ and for any two distinct sets $p_1, p_2$ in $\QQQ$ it holds $p_1\cap p_2 = c$.	
	Now let $(G,C',\chi',s,t,k-|c|)$ be the instance of \cmor\ such that $C'=C\setminus c$ and for every $v\in V(G)$, $\chi'(v) = \chi(v)\setminus c$ and let $\QQQ'=\{p\setminus c\mid p\in \QQQ\}$.

	\iflong
	\begin{claim}\label{clm:removingColors}
	\fi
	\ifshort
	\begin{claim}[$\proofMark$]\label{clm:removingColors}
	\fi 
		For all $p\in \QQQ$, $\QQQ'\setminus \{p\setminus c\}$ $(k-|c|)$-represents $\QQQ'$ w.r.t. $v$ in $(G,C',\chi',s,t,k-|c|)$ if and only if $\QQQ^v\setminus \{p\}$ $k$-represents $\QQQ^v$ w.r.t. $v$ in $(G,C,\chi,s,t,k)$.
	\end{claim}
	\iflong 
	\begin{cProof}
		Let $Q$ be a $v$-$t$ walk. Note that for any color set $p'$ a vertex $u$ is reachable from $s$ by $p'$ in $(G,C,\chi,s,t,k)$ if and only if it is reachable from $s$ by $p'\setminus c$ in $(G,C',\chi',s,t,k-|c|)$. Moreover, since $c\subseteq p''$ for every $p''\in \QQQ$ it holds $|p''\cup \chi(Q)|\le k$ if and only if $|(p''\setminus c)\cup \chi'(Q)|\le k-|c|$ and $p''\cap \chi(Q) = (p''\setminus c)\cap \chi'(Q)\cup (c\cap \chi'(Q))$. The proof then follows straightforwardly from the definition of $k$-representation w.r.t. $v$.
	\end{cProof}
	\fi 
	Removing the colors in $c$ from $G$ can result in an instance that is not irreducible w.r.t. color contraction. However, in our algorithm for color-disjoint case, we crucially rely on the fact that $G$ is irreducible w.r.t. color contraction. Now let $G_0=G$, $\chi_0=\chi'$, $s_0=s$, $t_0=t$, $v_0=v$ and for $i\ge 1$ let $(G_i,C,\chi_i, s_i, t_i, k-|c|)$ be an instance we obtain from $(G_{i-1},C,\chi_{i-1}, s_{i-1}, t_{i-1}, k-|c|)$ by a single color contraction of vertices $x_i$ and $y_i$ into a vertex $z_i$ and let $v_i=z_i$ if $v_{i-1}\in \{x_i, y_i\}$ and $v_i=v_{i-1}$ otherwise. 
	
	\iflong 
	\begin{claim}\label{clm:color_contractibility}
	\fi
	\ifshort  
	\begin{claim}[$\proofMark$]\label{clm:color_contractibility}
	\fi
		For all $p\in \PPP$, if the set $\PPP\setminus p$ $(k-|c|)$-represents $\PPP$ w.r.t. $v_i$ in $(G_i,C,$ $\chi_i,s_i,t_i,k-|c|)$, then $\PPP\setminus p$ $(k-|c|)$-represents $\PPP$ w.r.t. $v$ in $(G_{i+1},C,\chi_{i+1},s_{i+1},t_{i+1},k-|c|)$.
	\end{claim}
	\iflong 
	\begin{cProof}
		
		Let $Q=(u_1,\ldots, u_{|Q|})$ be a $v$-$t$ walk in $G_{i-1}$ such that $|p\cup\chi_{i-1}(Q)|\le k$ and $v_{i-1}$ is the only vertex on $Q$ reachable by $p$ from $s_{i-1}$. Also assume that there is no $s_{i-1}$-$v_{i-1}$ path $P'$ with $\chi_{i-1}(P')\subsetneq p$. Let $Q'=(u'_1,\ldots, u'_{|Q|})$ be a walk in $G_i$ such that if $u_j\notin \{x_i,y_i\}$, then $u'_j=u_j$ and $u'_j=z_i$ otherwise. Since $\chi_{i-1}(u_j)=\chi_{i}(u'_j)$ for all $j\in[|Q|]$, it follows that $\chi_{i-1}(Q)=\chi_{i}(Q')$, therefore $|p\cup \chi_i(Q')|\le k$. Moreover, from Observation~\ref{obs:colorContraction} follows that there is no $s$-$v$ path $P'$ in $G_i$ with $\chi_i(P')\subsetneq p$ and that $v_i$ is the only vertex on $Q'$ that is reachable from $s_i$ by $p$. Therefore, because $\PPP\setminus \{p\}$ $(k-|c|)$-represents $\PPP$ w.r.t. $v_i$ in $(G_i,C,\chi_i,s_i,t_i,k-|c|)$, there exists $p'\in \PPP\setminus \{p\}$ such that $|p'\cup \chi_i(Q')|\le k$, $p'\cap \chi_i(Q')\supseteq p\cap\chi_i(Q')$ and there is an $s$-$v$ path $P'$ with $\chi(P')=p'$. But then $|p'\cup \chi_{i-i}(Q)|\le k$, $p'\cap \chi_{i-1}(Q)\supseteq p\cap\chi_{i-1}(Q)$ and we can obtain an $s$-$v$ path $P''$ with $\chi(P'')=p'$ by taking $P'$ and replacing each vertex $w$ on $P'$ either by itself, if $w\in V(G_{i-1})$ or by one of the four subpaths ($(x_i)$, $(y_i)$, $(x_i, y_i)$, or $(y_i, x_i)$) depending on which of $x_i$, $y_i$ is adjacent to the predecessor and the successor of $z_i$ on $P'$.   
	\end{cProof}
	\fi 
	Let $(G_i,C,\chi_i,s_i,t_i,k-|c|)$ be the instance obtained from $(G,C',\chi',s,t,k-|c|)$ by repeating color contraction operation until $G_i$ is irreducible w.r.t. color contraction and let $v_i$ be the image of $v$. Since $G_i$ is irreducible w.r.t. color contraction, the sets in $\QQQ'$ are pairwise color-disjoint, and $|\QQQ'|=g(k)+1>g(k-|c|)$, we can use Lemma~\ref{lem:costructingDisjointFamily} to find in time polynomial in $|\QQQ'|+|V(G)|$ a set $p\in \QQQ'$ such that $\QQQ'\setminus\{p\}$ $(k-|c|)$-represents $\QQQ'$ w.r.t. $v_i$ in $(G_i,C,\chi_i,s_i,t_i,k-|c|)$. By Claim~\ref{clm:color_contractibility}, it follows that $\QQQ'\setminus\{p\}$ $(k-|c|)$-represents $\QQQ'$ w.r.t. $v$ in $(G,C',\chi',s,t,k-|c|)$ and by Claim~\ref{clm:removingColors} $\QQQ\setminus \{p\cup c \}$ $k$-represents $\QQQ$ in $(G,C,\chi,s,t,k)$. Finally, since for all $p'\in \PPP\setminus \QQQ$ is $p'\in \PPP\setminus \{p\cup c \}$ it follows that $\PPP\setminus \{p\cup c \}$ $k$-represents~$\PPP$. 
	\iflong 
	
	Note that finding a large sunflower, removing colors in $c$ from all vertices in $G$ and performing color contraction operation are all polynomial time procedures and we cannot repeat the color contraction operation more than $|V(G)|$ many times, as each time the number of vertices in graph is reduced by one. Hence the above described algorithm runs in time polynomial in $|\PPP|+|V(G)|.$
	\fi 
\end{proof}

\bibliography{ref}

\end{document}